\documentclass[letterpaper, 10pt, conference]{ieeeconf} 

\IEEEoverridecommandlockouts                              
\usepackage{float}
\usepackage{amsmath}
\usepackage{amssymb}
\usepackage{dsfont}
\usepackage{graphicx}
\usepackage{mathtools}
\usepackage{theorem}
\usepackage{pstricks}
\usepackage{epsfig}
\usepackage{pst-grad} 
\usepackage{pst-plot} 
\usepackage{bbm}
\usepackage{float}
\usepackage{footnote}
\usepackage{cite}

{\theorembodyfont{\itshape}\newtheorem{theorem}{Theorem}}
{\theorembodyfont{\itshape}\newtheorem{proposition}{Proposition}}
{\theorembodyfont{\itshape}}
{\theorembodyfont{\itshape}\newtheorem{corollary}{Corollary}}
{\theorembodyfont{\upshape}}
{\theorembodyfont{\itshape}}
{\theorembodyfont{\upshape}\newtheorem{remark}{Remark}}
{\theorembodyfont{\upshape}}
{\theorembodyfont{\upshape}}

\newcommand{\grond}[2]{\frac{\partial #1}{\partial #2}}

\newcommand\numberthis{\addtocounter{equation}{1}\tag{\theequation}}

\newcommand{\R}{\mathbb{R}}

\newcommand{\w}{\omega}

\newcommand{\pu}{\mathrm{pu}}

\begin{document}

\title{\LARGE \bf 
Nonlinear Analysis of an Improved Swing Equation}

\author{Pooya~Monshizadeh$^{1}$, Claudio De Persis$^{2}$, Nima Monshizadeh$^{2}$, and Arjan van der Schaft$^{1}$
\thanks{$^{1}$Pooya Monshizadeh and Arjan van der Schaft are with the Johann Bernoulli Institute for Mathematics and Computer Science, University of Groningen, 9700 AK, the Netherlands,
        {\tt\small p.monshizadeh@rug.nl, a.j.van.der.schaft@rug.nl}}%
\thanks{$^{2}$Claudio De Persis and Nima Monshizadeh are with the Engineering and Technology Institute, University of Groningen, The Netherlands, University of Groningen, 9747 AG, the Netherlands,
        {\tt\small n.monshizadeh@rug.nl, c.de.persis@rug.nl}}%
}

\maketitle

\begin{abstract}
In this paper, we investigate the properties of an improved swing equation model for synchronous generators. This model is derived by omitting the main simplifying assumption of the conventional swing equation, and requires a novel analysis for the stability and frequency regulation. We consider two scenarios. First we study the case that a synchronous generator is connected to a constant load. Second, we inspect the case of the single machine connected to an infinite bus. Simulations verify the results.
\end{abstract}

\section{Introduction}
Driven by environmental and technical motivations, restructuring the classical power networks has been under vast attention during the recent decades. Among the goals are decreasing energy losses by moving towards distributed generation and preventing fault propagation through building up smart \textit{microgrids}. Microgrids are small power network areas which can be seen as single entities from the large power grids. In such small-scale network, the energy consumption and production uncertainty increases to a great extent according to the fewer number of consumers and unpredictable power injections of renewables such as wind turbines and solar panels. Designing controllers for the electrical sources under such perturbations and abrupt power variations, calls for rethinking about the accuracy of the models that were (mostly) valid for the classical electrical systems.
\\
The main power sources in electrical networks are synchronous generators. Despite the extensive advances in extracting energy from renewables, these machines are still the main supplier of the implemented microgrids (see e.g. \cite{Shahidehpour2013}). Under unpredictable changes in loads and renewable generations, and to assure of the stability and frequency regulation of the (micro-) grid, it appears crucial to investigate the accuracy of the electrical and dynamical models of these machines.
\\
While a large number of articles have exploited the classical \textit{swing equation} as the model for the synchronous generator, a few have recently brought up doubts about its accuracy and validity, and proposed models of higher accuracy \cite{Zhou2009,Arjan2013,Caliskan2014,Tabuada2015,Weiss20141,Weiss20142,Tjerk2016}. Although some promising models are provided (see e.g. the exact 8-dimensional model in \cite{Arjan2013} and the integro-differential model in \cite{Weiss20141}), the use of such models in power networks is rather complicated.
\\
In this paper, we provide a nonlinear analysis for an improved yet easy-to-use swing equation to fulfill higher accuracy in analysis and modeling of the synchronous generator that can also be exploited in small-scale networks. We first explicate the contradicting assumptions in obtaining the conventional swing equation and elaborate deriving the improved model. In Section III, as the first scenario, we look into the properties of the model when the generator is connected to a constant load and provide an estimate of the region of attraction through nonlinear analysis. Further in Section IV, the second scenario which is mainly referred to as \textit{Single Machine Infinite Bus} (SMIB) is investigated. We provide analytical estimates of the region of attraction for both conventional and improved swing models in this case. Finally, simulations are provided as a verification of the results.
\section{Conventional and Improved Swing Equations}
The mechanical dynamics of the synchronous generator reads as 
\begin{equation}\label{swing2}
J \dot{\w }+D_d (\w-\w^*) =\tau_m-\tau_e  \; \text{,} 
\end{equation} 
where $J \in \R^+$ is the total moment of inertia of the turbine and generator rotor ($kg\,m^2$),  $\w \in \R^+$ is the rotor shaft velocity (mechanical $rad/s$), $\w^* \in \R^+$ is the angular velocity associated with the nominal frequency ($(2\pi)60\mathrm{Hz}$), $\tau_m \in \R^+$ is the net mechanical shaft torque ($N\, m$), $\tau_e \in \R^+$ is the counteracting electromagnetic torque ($N\, m$), and $D_d \in \R^+$ is the damping-torque coefficient ($N\,m\,s$) \footnote{Despite the general misuse, $D_d$ accounts for the damping torque generated by the damper windings (amortisseur) only when the generator is connected to an infinite bus \cite{Machowski2009}. In the case of a single generator connected to a load, the damping term refers to a proportional torque governor.}.
The mechanical rotational loss due to friction is ignored. Bearing in mind that $\tau_m=\frac{P_m}{\w}$ and $\tau_e=\frac{P_e}{\w}$ we can model the synchronous generator as
\begin{equation}\label{swing3}
J \w \dot{\w }+D_d \w (\w-\w^*) =P_m-P_e  \;  \text{,} 
\end{equation} 
where $P_m$ and $P_e$ are the mechanical (input) and electrical  (output) power respectively.
\\
In the conventional swing equation, $J\w$ and $D_d\w$ are approximated with constants $M=J\w^*$ and $A=D_d\w^*$.
\begin{equation}\label{swing1}
M\dot{\w }+A(\w-\w^*) =P_m-P_e  \; \text{,} 
\end{equation}
where $M \in \mathbb{R}$ is the angular momentum and $A\in \mathbb{R}$ is the (new) damping coefficient.
In other words, the swing equation is derived supposing that $\w=\w^*$. Such an assumption is in contradiction with the proof of stability and frequency regulation of the synchronous generator. Note that these machines posses high moment of inertia ($J$), and hence small deviations from the nominal frequency are magnified by the term $J\w$ in the dynamics. Therefore we adhere to the equation \eqref{swing3} and endeavor to investigate the stability and frequency regulation through a nonlinear approach.  This model has been first suggested by \cite{Zhou2009}, where however, the stability of the synchronous generator (connected to an infinite bus) is analyzed with the small-signal (linearization) analysis. Consistently with \cite{Zhou2009}, we refer to \eqref{swing3} as the \textit{improved swing equation}.
\section{Synchronous Generator \\Connected to a Constant Load} \label{SIM}
In this scenario, we assume that both the injected and extracted power ($P_m$ and $P_e$) are constant. 
\subsection{Stability}
Let $\bar\omega\in \R$ be an equilibrium of \eqref{swing3}. Then, clearly
\begin{equation}\label{equil}
D_d \bar{\w} (\bar{\w}-\w^*) =P_m-P_e  \; \text{.} 
\end{equation} 
The equation above admits a real solution only if
\begin{align}\label{delta}
P_m-P_e>-\frac{1}{4} D_d \w^{*2} \;,
\end{align}
which will be a standing assumption in this section.
Under this assumption, we obtain the following two equilibria for system \eqref{swing3}:
\begin{align}\label{equib}
& \bar{\w}_s=\frac{\w^*+\sqrt{\Delta}}{2} ,\quad
 \bar{\w}_u=\frac{\w^*-\sqrt{\Delta}}{2} 
\; \text{,} 
\end{align}
where 
\begin{equation}\label{e:Delta}
\Delta:=\w^{*2}+4\frac{P_m-P_e}{D_d}.
\end{equation}
Note that $\Delta>0$ by \eqref{delta}. Observe that
\begin{align}\label{roots}
\bar{\w}_s&>\bar{\w}_u \; \text{,}\\
\label{rootseq}
\bar{\w}_u&=\w^*-\bar{\w}_s \; \text{.}
\end{align}
First, we show that $\bar{\w}_s$ is locally stable and $\bar{\w}_u$ is locally unstable. 
For the moment, we assume that the set $\R^+$ is positive invariant for the system \eqref{swing3}, and $\omega(0)>0$. We will relax this assumption later, once a Lyapunov argument is provided. This means that $\omega^{-1}$ is well defined for the interval of the definition of the solutions. Now, let dynamics \eqref{swing3} be rewritten as 
\[
\dot\w=f(\w)
\]
where
\[
f(\w)=\frac{1}{J}\big(-D_d (\w-\w^*) +\frac{P_m-P_e}{\w} \big).
\]
We have
\begin{align*}
\grond{f}{\w}=\frac{1}{J}(-D_d -\frac{P_m-P_e}{\w^2}) \; \text{.} 
\end{align*}
By using \eqref{equil} and \eqref{rootseq}, we find that
\begin{align*}
\grond{f}{\w}&=\frac{1}{J}(-D_d -\frac{D_d \bar{\w}_s (\bar{\w}_s-\w^*)}{\w^2})
\\
&=\frac{D_d}{J}(\frac{\bar{\w}_s \bar{\w}_u}{\w^2}-1) \numberthis \label{J}.
\end{align*}
Bearing in mind inequality \eqref{roots}, it is easy to check that $\grond{f}{\w}> 0$ around $\w=\bar{\w}_u$, and $\grond{f}{\w}< 0$ around $\w=\bar{\w}_s$. Hence $\w=\bar{\w}_u$ is repulsive. Now, the following theorem addresses the stability of the equilibrium $\bar\w$ in \eqref{equib}.
\begin{figure}[t] 
\includegraphics[width=8.5cm]{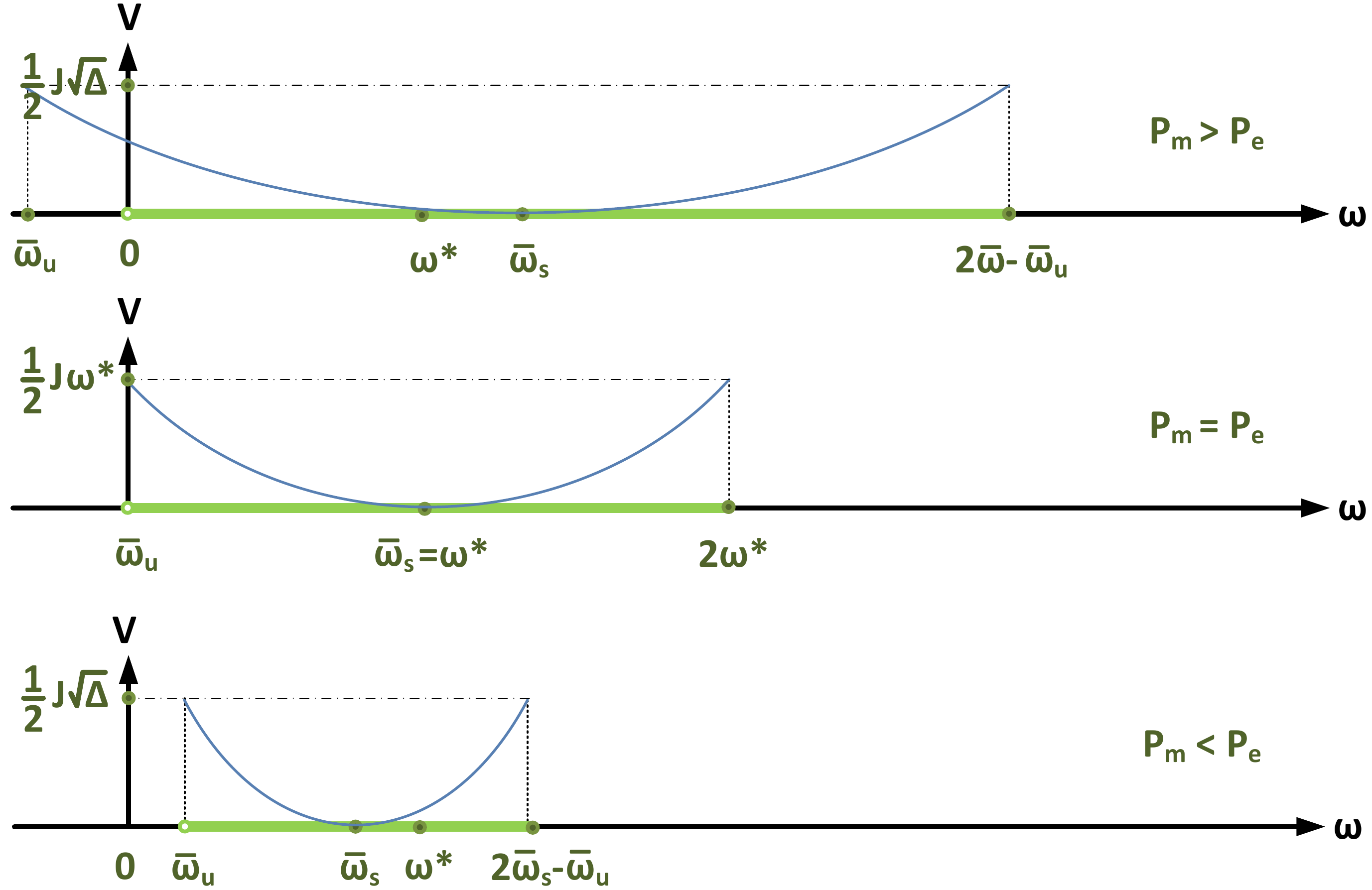} 
\caption{Region of attraction for system \eqref{swing3} is estimated by $\Omega_s$.} \label{Omega}
\end{figure}
\begin{theorem}\label{th:stability}
Consider the candidate Lyapunov function $$V(\w)=\frac{1}{2}J(\w-\bar{\w}_s)^2$$
with $\bar\omega_s$ given by \eqref{equib}. Let $\Delta$ be given by \eqref{e:Delta} and assume that \eqref{delta} holds. Then the solutions of the system \eqref{swing3} starting from any initial condition in the set $\Omega_s = \{\w \in \R^+:\,V(\w)\leq\frac{1}{2}J\Delta\}$ except for $\w(0)=\bar{\w}_u$ converge asymptotically to the equilibrium $\w=\bar\w_s$. 
\end{theorem}
\begin{proof}
We have
\begin{align*}
\dot{V}&=\w^{-1}(\w-\bar{\w}_s)\Big(P_m-P_e-D_d\w(\w-\w^*)\Big).
\end{align*}
The fact that $\w$ stays away from zero will be made clear later.
According to \eqref{equil} this leads to
\begin{align*}
\dot{V}&=\w^{-1}(\w-\bar{\w}_s)\Big(D_d \bar{\w}_s (\bar{\w}_s-\w^*)-D_d\w(\w-\w^*)\Big)
\\
&=-\w^{-1}D_d(\w-\bar{\w}_s)\Big(\w^2-\w\w^*+\bar{\w}_s\w^*-\bar{\w}_s^2\Big)
\\
&=-D_d\w^{-1}(\w-\bar{\w}_s)\Big((\w-\bar{\w}_s)(\w+\bar{\w}_s-\w^*)\Big)
\\
&=-D_d(\w-\bar{\w}_s)^2(1-\frac{\bar{\w}_u}{\w})\;\text{,} 
\end{align*}
where we have used $\bar{\w}_u=\w^*-\bar{\w}_s$. Hence $\dot{V}$ is negative semi-definite on the set  $\Omega'_s = \{\w \in \R^+:\,\w\geq\bar{\w}_u\}$. Below, we show that $\Omega_s \subseteq \Omega'_s$.
\\
For any $\w\in\Omega_s$ we have
\begin{align*}
& \frac{1}{2}J(\w-\bar{\w}_s)^2\leq\frac{1}{2}J(\w^{*2}+4\frac{P_m-P_e}{D_d})
\end{align*}
which yields
\[
(\w-\bar{\w}_s)^2\leq(2\bar{\w}_s-\w^*)^2
\]
by using \eqref{equib}. By \eqref{delta} we have $2\bar{\w}_s-\w^*\geq0$. Therefore
\begin{align*}
 \w^*-2\bar{\w}_s\leq\w-\bar{\w}_s\leq2\bar{\w}_s-\w^*
\end{align*}
and hence
\begin{align}
 \bar{\w}_u\leq\w\leq2\bar{\w}_s-\bar{\w}_u \label{bounded}
 \; \text{.}
\end{align}
The left hand side of \eqref{bounded} gives $\Omega_s \subseteq \Omega'_s$. 
\\
Now we show that $\w$ remains in the positive half-line $\R^+$. By \eqref{equib} and \eqref{e:Delta}, in the case $P_m<P_e$, we have $\bar{\w}_u>0$ and together with \eqref{bounded} we obtain $\w>0$. Therefore in this case, the point $\w=0$ does not belong to the interval specified by \eqref{bounded}. In cases $P_m > P_e$ or $P_m = P_e$, the point $\w=0$ lies within the boundaries \eqref{bounded}. Figure \ref{Omega} illustrates these situations. Since $(\w-\bar{\w}_s)^2=2JV(\w)$, and $V(\w)$ is strictly decreasing in $\Omega_s$, $(\w-\bar{\w}_s)^2$ is decreasing as well, showing that $\w$ stays away from  $\w=0$ as time goes by. This proves that $\w(t)>0,\;\forall t>0$. 
\\
Notice that $\Omega_s$ is closed and bounded by \eqref{bounded}. 
By invoking LaSalle's invariance principle, the solutions of the system starting in $\Omega_s$ asymptotically converge to the set of points where 
\begin{align} \label{vdot}
-D_d(\w-\bar{\w}_s)^2(1-\frac{\bar{\w}_u}{\w})=0  \; \text{.}
\end{align}
Equality \eqref{vdot} results in $\w=\bar{\w}_s$, which completes the proof.
\end{proof}
\begin{remark}
Note that the Lyapunov function $V(\w)=\frac{1}{2}J(\w-\bar{\w}_s)^2$ is a shifted version of the actual kinetic energy of the physical system, whereas the commonly used Lyapunov function for the swing equation, $\frac{1}{2}M(\w-\bar{\w}_s)^2$ equivalent to $\frac{1}{2}J\w^\ast(\w-\bar{\w}_s)^2$, does not have a physical interpretation.
\end{remark}
\subsection{Considering Mechanical Losses}
We add the viscous damping coefficient $D_m$ accounting for the mechanical losses ($N\,m\,s$), modifying the dynamics \eqref{swing2} as
\begin{equation}\label{swing5}
J \dot{\w }+D_m \w+D_d (\w-\w^*) =\tau_m-\tau_e  \; \text{,} 
\end{equation} 
and similarly \eqref{swing3} is modified as
\begin{equation}\label{swing6}
J \w \dot{\w }+D_m \w^2+D_d \w (\w-\w^*) =P_m-P_e  \;  \text{,} 
\end{equation}
which can be rewritten as 
\begin{equation}\label{swing7}
J \w \dot{\w }+(D_m+D_d) \w (\w-\frac{D_d}{D_m+D_d}\w^*) =P_m-P_e  \;  \text{.} 
\end{equation} 
Defining $D:=D_m+D_d$ and $\tilde{\w}^*:=\frac{D_d}{D_m+D_d}\w^*$ we obtain
\begin{equation}\label{swing8}
J \w \dot{\w }+D \w (\w-\tilde{\w}^*) =P_m-P_e  \;  \text{,} 
\end{equation} 
which is analogous to the improved swing equation \eqref{swing3}. Therefore the stability results also extends to the case with mechanical losses.
\subsection{Incremental Passivity Property}
To facilitate the control design, in this section we investigate the incremental passivity of the system associated to \eqref{swing3} (see Remark \ref{ipremark} later on)
\begin{align} \label{ip}
\begin{aligned}
& J \w \dot{\w }+D_d \w (\w-\w^*) =u+P_m-P_e  
\\
& y=\frac{\w-\w^*}{\w}   \; \text{.}
\end{aligned}
\end{align}
Let the triple $(\bar u, \bar \w, \bar y)$ be an input-state-output solution of \eqref{ip} with
\begin{align} 
\begin{aligned}
& D_d \bar\w (\bar{\w}-\w^*) =\bar u+P_m-P_e  
\\
& y=\frac{\bar{\w}-\w^*}{\bar\w}   \; \text{.}
\end{aligned}
\end{align}
Assume that 
\begin{align}\label{asmp}
\bar{u}+P_m-P_e>-\frac{1}{4} D_d \w^{*2} \;.
\end{align}
Then the dynamics \eqref{ip} possesses the equilibria
\begin{align}\label{equib3}
\bar{\w}_s=\frac{\w^*+\sqrt{\Delta}}{2} , \quad \bar{\w}_u=\frac{\w^*-\sqrt{\Delta}}{2}
\end{align}
where with a little abuse of the notation
\begin{align}\label{e:Delta2}
\Delta:=\w^{*2}+4\frac{\bar{u}+P_m-P_e}{D_d}.
\end{align}
Now, we have the following proposition:
\begin{proposition}\label{l:ip}
Consider the candidate Lyapunov function 
$$
W(\w)=\frac{1}{2}J(\w-\bar{\w}_s)^2\frac{\w^*}{\bar{\w}_s}
$$ and assume that \eqref{asmp} holds. Then, $W(\w)$ computed along any solution to \eqref{ip} satisfies
\begin{align}\label{ipcondition}
\dot{W}<(y-\bar{y})(u-\bar{u})\;,
\end{align} 
as long as the solution stays within the set $\Omega_k = \{\w \in \R^+:\,W(\w)\leq\frac{1}{2}J\Delta\frac{\w^*}{\bar{\w}_s}\}$. 
This amounts to an incremental passivity property with respect to $(\bar u, \bar\w, \bar y)$.
\end{proposition}
\begin{proof}
We have
\begin{align*}
\dot{W}=&(\frac{1}{\w})(\frac{\w^*}{\bar{\w}_s})(\w-\bar{\w}_s)
\Big(u+P_m-P_e-D_d\w(\w-\w^*)
\\
&-\bar{u}-P_m+P_e+D_d \bar{\w}_s (\bar{\w}_s-\w^*)\Big)
\\
=&\Big(\frac{\w^*(\w-\bar{\w}_s)}{\w\bar{\w}_s}\Big)(u-\bar{u})
\\
&-\Big(\frac{\w^*(\w-\bar{\w}_s)}{\w\bar{\w}_s}\Big)\Big(D_d\w(\w-\w^*)-D_d \bar{\w}_s (\bar{\w}_s-\w^*)\Big)
\\
=&(\frac{\w-\w^*}{\w}-\frac{\bar{\w}_s-\w^*}{\bar{\w}_s})(u-\bar{u})
\\&-D_d\Big(\frac{\w^*(\w-\bar{\w}_s)}{\w\bar{\w}_s}\Big)\Big((\w-\bar{\w}_s)(\w+\bar{\w}_s-\w^*)\Big)
\\
=&(y-\bar{y})(u-\bar{u})
\\&-D_d(\w-\bar{\w}_s)^2(1-\frac{\bar{\w}_u}{\w})(\frac{\w^*}{\bar{\w}_s})\;\text{.}  \numberthis \label{wdot}
\end{align*}
Hence the inequality \eqref{ipcondition} holds as long as the solutions evolve in the set $\Omega'_k = \{\w \in \R^+:\,\w\geq\bar{\w}_u\}$. The fact that $\Omega_k \subseteq \Omega'_k$ follows analogously to the proof of Theorem \ref{th:stability}.
\end{proof}
\begin{remark}\label{ipremark}
Defining $y=\w^{-1}(\w-\w^*)$ as in \eqref{ip}, results in a dimensionless output. This is consistent with the inequality of incremental passivity property \eqref{ipcondition}, since $yu$ has the dimension of physical power. This does not hold for the conventional swing equation where $\w$ is taken as the output and thus $yu$ does not have a meaningful physical dimension. 
\end{remark}
\subsection{Frequency Regulation} \label{FR}
In this section we investigate the frequency regulation of the system by an integral controller. Motivated by Proposition \ref{l:ip}, we propose the controller for system \eqref{ip}
\begin{align} \label{i}
\begin{aligned}
& \dot{\xi}=y=\frac{\w-\w^*}{\w}
\\
& u=-\xi
\; \text{.} 
\end{aligned}
\end{align}
Note that for the purpose of frequency regulation, the solution (equilibrium) of interest for the system \eqref{ip}, \eqref{i}, is given by 
\begin{equation}\label{e:desired}
\bar u=-\bar \xi=-(P_m-P_e), \quad \bar\w=\omega^\ast, \quad \bar y=0,
\end{equation}
which trivially satisfies the inequality \eqref{asmp}. In addition, for the solution above,  the quantities in \eqref{equib3} and \eqref{e:Delta2} are computed as
\[
\bar{\w}_s=\w^*, \quad \bar{\w}_u=0, \quad \Delta=\w^{\ast 2}.
\]
Therefore, the dissipation equality \eqref{wdot} in this case reduces to
\begin{align}\label{e:dissipate}
\dot{W}=&(y-\bar{y})(u-\bar{u})-D_d(\w-\bar{\w}_s)^2,
\end{align}
as long as the solutions stay within the set $\Omega_k = \{\w \in \R^+:\,W(\w)\leq\frac{1}{2}J\w^{*2}\}$ which confines $\w$ within the interval $\w\in[0\;\,2\w^*]$. 
Now, the following theorem establishes the convergence of the solutions to the desired equilibrium associated with the nominal frequency regulation.
\begin{theorem}
Consider the candidate Lyapunov function $$U(\xi, \w)=W_c(\xi)+W(\w)$$ where $W_c=\frac{1}{2}(\xi-\bar{\xi})^2\;
$ and $W=\frac{1}{2}J(\w-\w^*)^2$. Then the solutions of the closed loop system \eqref{ip}, \eqref{i}, starting from any initial condition in the set $\mathcal{O} = \{(\xi,\w) \in \R^2:\,U\leq\frac{1}{2}J\w^{*2}\}$ and different from $(\xi,\w)=(P_m-P_e,0)$ converge asymptotically to the equilibrium 
$(\xi,\w)=(P_m-P_e,\w^*)$.
\end{theorem}
\begin{proof}
First note that
\[ 
\dot{W}_c=\w^{-1}(\w-\w^*)(\xi-\bar{\xi})
\]
Hence, by \eqref{e:dissipate} we find that
\begin{align}\label{e:dotU}
\dot{U}=&-D_d(\w-\w^*)^2.
\end{align}
as long as the solutions belong to the set $\mathcal{O}$. As the right hand side of the above equality is nonpositive, we conclude that the set $\mathcal{O}$ is forward invariant along the solution $( \xi,\w)$, and thus \eqref{e:dotU} is valid for all time. 
Now, observe that $U$ is radially unbounded and has a strict minimum at $(\xi,\w)=(P_m-P_e,\w^*)$. 
Then by invoking LaSalle's invariance principle, the solutions converge to the largest invariant subset of $\mathcal{O}$ in which 
\begin{align*}
\w&=\w^\ast\\
\bar{u}&=P_e-P_m \;\text{.}
\end{align*}
The latter shows that all the solutions on the invariant set $\mathcal{O}$, and not starting from $(\xi,\w)=(P_m-P_e,0)$ converge asymptotically to the equilibrium $(\xi,\w)=(P_m-P_e,\w^*)$. Note that the solutions stay away from the point $(\xi,\w)=(P_m-P_e,0)$, since as $(\xi,\w)\rightarrow(P_m-P_e,0^+)$, we have $\dot{\w}\rightarrow+\frac{D_d}{J}\w^*$. We can rewrite the set  $\mathcal{O}$ as
\begin{equation}\label{e:oval}
\big(\frac{\w-\w^*}{\w^*}\big)^2+\big(\frac{\xi-\bar{\xi}}{\sqrt{J}\w^*}\big)^2\leq1.
\end{equation}
Such a set guarantees $\w$ to remain in the positive half-line $\R^+$ for $\w$ (See figure \ref{O}).
\end{proof}
\begin{figure}[t]
	\centering
	\includegraphics[width=3cm]{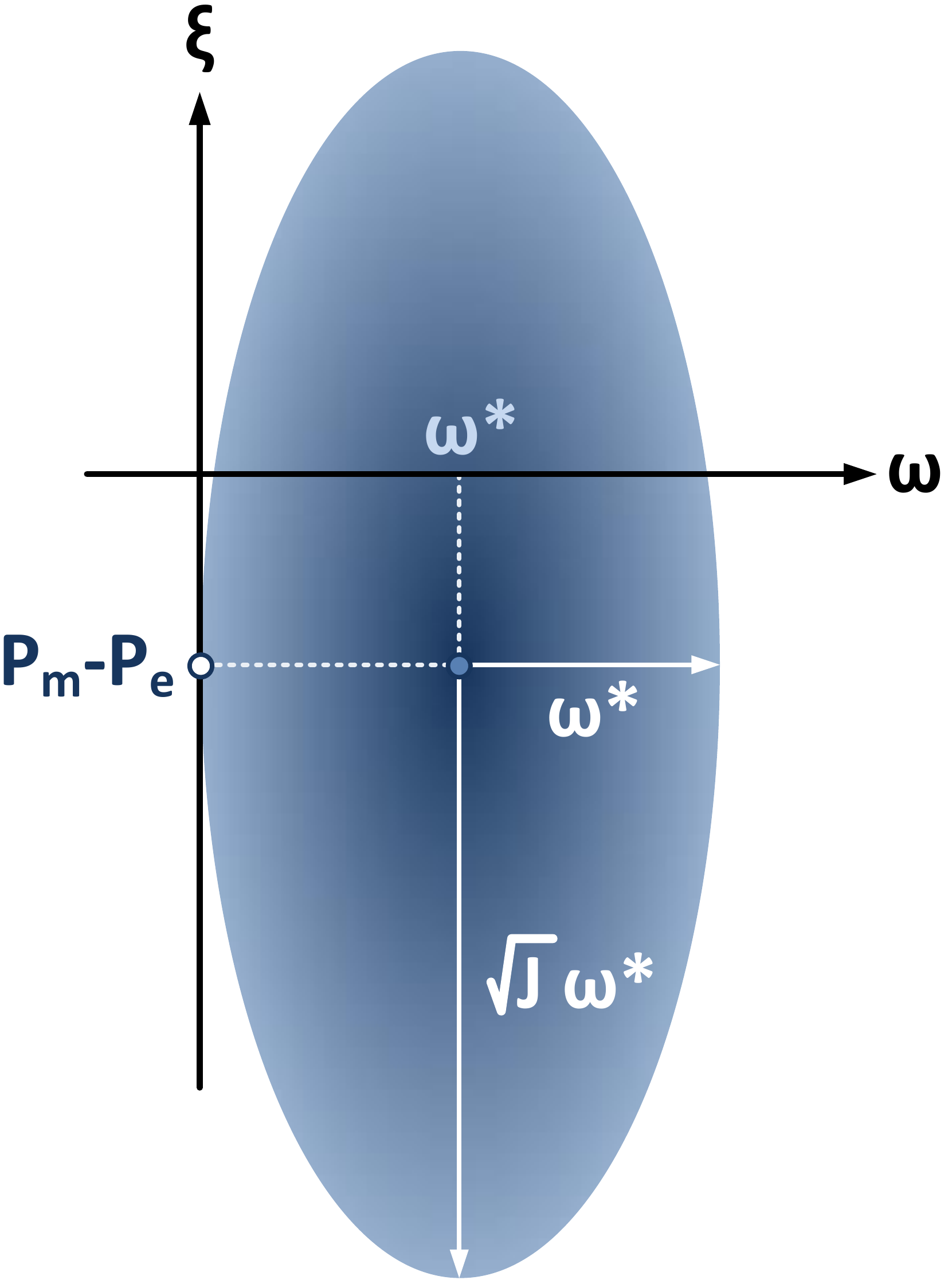} 
	\caption{Region of attraction for the system \eqref{ip}, \eqref{i} is estimated with the oval characterized by \eqref{e:oval}. } \label{O}
\end{figure}
\section{Synchronous Generator \\Connected to an Infinite Bus}
In this scenario we consider the case of SMIB, where the synchronous generator is connected to an infinite bus through an inductive link. The infinite bus is a node with fixed voltage and frequency (60Hz). 
\subsection{Stability of the Improved Swing Model}
It can be shown that the power delivered from the synchronous generator to the bus is $\gamma \sin(\delta)$ where $\delta$ is the voltage angle relative to the infinite bus, and $\gamma:=\frac{V_gV_b}{X}$. Here, the voltages ($V_g,V_b$) and the reactance of the line ($X$) are assumed to be constant. The dynamics of an SMIB system modeled with the improved swing equation \eqref{swing3} is
\begin{align}\label{swing4}
\begin{aligned}
& \dot{\delta}=\w-\w^*
\\
& J \w \dot{\w }+D_d \w (\w-\w^*) =P_m-\gamma \sin(\delta)  \;  \text{,} 
\end{aligned}
\end{align} 
where the mechanical input $P_m>0$ is considered constant. It is easy to check that $(\delta,\w)=(\mathrm{arcsin} \frac{P_m}{\gamma},\w^*)$ is the equilibrium of the system \eqref{swing4}.
\begin{theorem} \label{ThImproved}
Consider the candidate Lyapunov function $$V(\delta,\w)=V_k(\w)+V_p(\delta)\;,$$where $$V_k(\w)=\frac{1}{2}J(\w-\bar{\w}_s)^2$$and$$V_p(\delta)=\frac{\gamma}{\w^*}[-\cos\delta+\cos\bar{\delta}-(\delta-\bar{\delta})\sin\bar{\delta})]$$ are associated with the kinetic and potential energy and $\bar{\delta}=\mathrm{arcsin} \frac{P_m}{\gamma}$. Assume that 
\begin{align}\label{c:ib1}
\w^*>\sqrt{\frac{\gamma}{D_d}}
\end{align} 
and 
\begin{align}\label{c:ib2}
\frac{P_m}{\gamma}<\frac{2}{\pi}\;.
\end{align}
Let $c:=\min(c_k,c_p)$, where $c_k:=\frac{1}{2}J(\w^*-\frac{\gamma}{D_dw^*})^2$ and $c_p:=V_p(\frac{\pi}{2})$. Then the solutions of the system \eqref{swing4} starting from any initial condition in the set $\Omega = \{(\delta,\w) \in [-\pi,\pi]\times\R:\,V(\delta,\w)\leq c\}$ converge asymptotically to the equilibrium $(\delta,\w)=(\mathrm{arcsin} \frac{P_m}{\gamma},\w^*)$.
\end{theorem}
\begin{proof}
Observe that $V_p(\delta)$ has a minimum in $\bar{\delta}$ and is convex within the set $\Omega_p=\{\delta \in [-\frac{\pi}{2},\frac{\pi}{2}]\} $. We here show that $\Omega \subset \Omega_p$. Having $V_k(\w)>0$, the inequality $V(\delta,\w)<c$ results in $V_p(\delta)<c_p$ and reads as
\begin{align} \label{numerical}
\cos\delta>(\frac{\pi}{2}-\delta)\sin\bar{\delta} \;.
\end{align}
The inequality \eqref{numerical} contains a unique subset of $\Omega_p$. More precisely, under the criterion $0<\bar{\delta}<\mathrm{arcsin}\frac{2}{\pi}$ as a result of \eqref{c:ib2}, there exists a $\delta^-$ s.t. \eqref{numerical} holds for all $\delta\in [\delta^- \; \frac{\pi}{2}]$. Figure \ref{fig:numerical} shows an example interval for $\delta$ that satisfies \eqref{numerical} and consequently the energy function $V_P$ remains convex. 
Now, it remains to prove that $\dot{V}<0$. we have
\begin{align*}
\dot{V}=&-D_d(\w-\w^*)^2-\frac{\gamma}{\w}[(\w-\w^*)(\sin\delta-\sin\bar{\delta})]
\\
&+\frac{\gamma}{\w^*}[(\w-\w^*)(\sin\delta-\sin\bar{\delta})]
\\
=&-D_d(\w-\w^*)^2+\frac{\gamma}{\w\w^*}(\w-\w^*)^2(\sin\delta-\sin\bar{\delta})
\\
=&-(\w-\w^*)^2\big(D_d-\frac{\gamma}{\w\w^*}(\sin\delta-\sin\bar{\delta})\big)\;.
\end{align*}
Bearing in mind that $\delta$ is confined s.t. $\delta-\bar{\delta}<\frac{\pi}{2}$, $\dot{V}<0$ on the set $\Omega_k = \{\w \in \R:\,\w>\frac{\gamma}{D_d\w^*}\}$. Here we show that $\Omega \subset \Omega_k$. In the set $\Omega$ we have $V(\delta,\w)<c$ which leads to $V_k(\w)<c_k$. Therefore
\begin{align*}
&(\w-\w^*)^2<(\w^*-\frac{\gamma}{D_dw^*})^2 \;,
\end{align*}
and since $\w^*>\sqrt{\frac{\gamma}{D_d}}$,
\begin{align*}
&\frac{\gamma}{D_dw^*}-\w^*<\w-\w^*<\w^*-\frac{\gamma}{D_dw^*} \;,
\end{align*}
thus 
\begin{align}\label{final}
&\frac{\gamma}{D_dw^*}<\w<2\w^*-\frac{\gamma}{D_dw^*} \;.
\end{align}
The left hand side of the inequality \eqref{final} shows that $\Omega \subset \Omega_k$. 
\\
Observe that $U$ has a strict minimum at $\w=\w^*$ and $\delta=\bar{\delta}$, and the solutions are bounded to $\frac{\gamma}{D_dw^*}<\w<2\w^*-\frac{\gamma}{D_dw^*}$ and $\delta^-<\delta<\frac{\pi}{2}$. By invoking LaSalle's invariance principle, the solutions converge to the largest invariant subset of $\Omega$ for \eqref{swing4} s.t. $\w=\w^*$. On this set, the solution to \eqref{swing4} satisfy
\begin{align*}
& \delta=\mathrm{arcsin}\frac{P_m}{\gamma} \;\text{.}
\end{align*}
This shows that all the solutions on the invariant set converge asymptotically to the equilibrium ($\mathrm{arcsin} \frac{P_m}{\gamma},\w^*$). This completes the proof.
\end{proof}
\begin{remark}
The assumption $\frac{P_m}{\gamma}<\frac{2}{\pi}$ results in $\bar{\delta}<\mathrm{arcsin}\frac{2}{\pi}$ which means that at steady state the angle between the synchronous generator and the infinite bus should be less than about $40^\circ$. In practice, the steady state angle is much lower even for large machines (see e.g. \cite{Zhao1995}). Note that this assumption is made merely for characterizing the region of attraction, otherwise it is not necessary for the proof of local stability since $V$ is convex around the equilibrium ($\frac{\partial^2 V}{\partial \delta^2}|_{\delta=\bar{\delta}}=\cos\bar{\delta}>0$).
\end{remark}
\begin{figure}
	\centering
	\includegraphics[width=6.6cm]{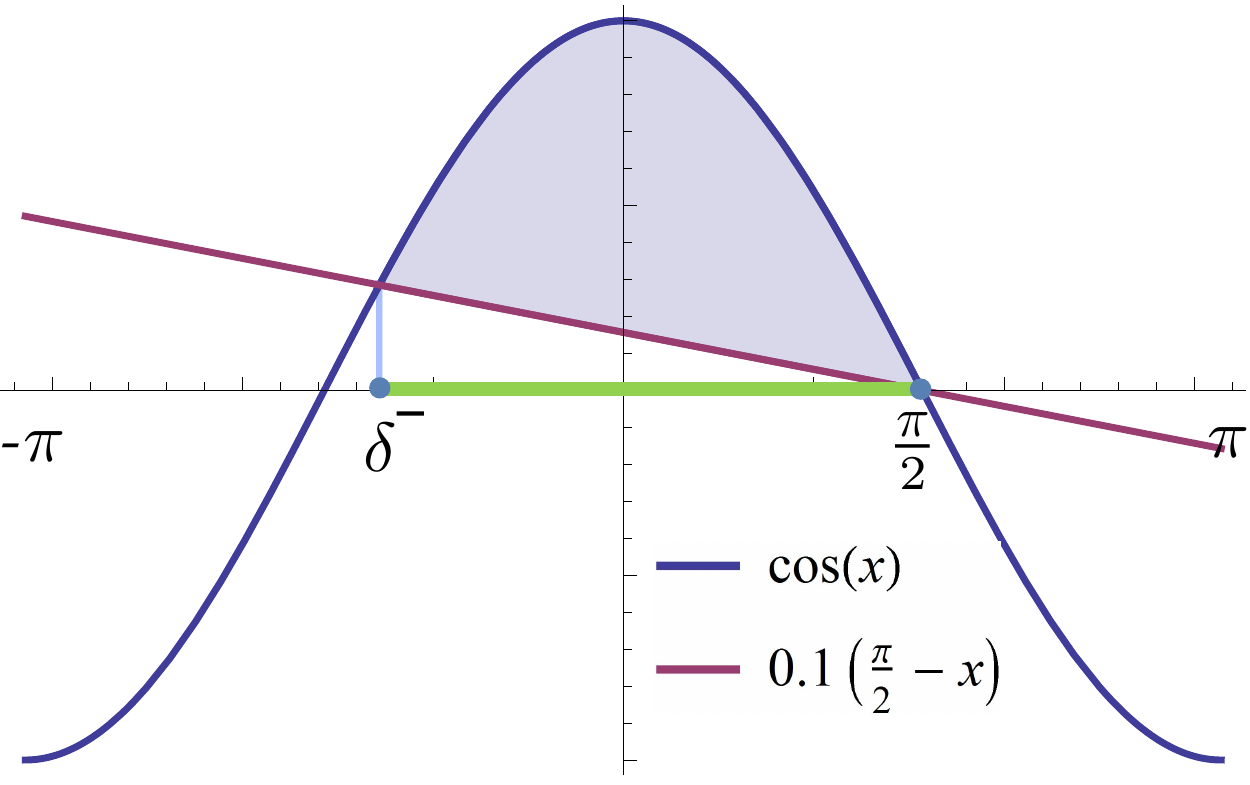}
	\caption{The estimate of the system \eqref{swing4} region of attraction ($V_p(\delta)<c_p$) for $\bar{\delta}=\mathrm{arcsin}0.1$ confines $\delta$ in $[-0.78\pi \;\; +0.5\pi]$ (green). On such an interval $V_p(\delta)$ remains positive and convex.}
	\label{fig:numerical}
\end{figure}
\subsection{Comparison with the Swing Equation}
Here we compare the results of the improved model with the swing equation. The dynamics of a synchronous generator modeled with the swing equation and connected to an infinite bus is
\begin{align} \label{swingib}
\begin{aligned}
	& \dot{\delta}=\w-\w^*
	\\
	& M \dot{\w }+A (\w-\w^*) =P_m-\gamma \sin(\delta)  \;  \text{,} 
\end{aligned}
\end{align} 
where the mechanical input $P_m>0$ is considered constant. It is straightforward to see that the equilibrium of the system \eqref{swingib} is $(\delta,\w)=(\mathrm{arcsin} \frac{P_m}{\gamma},\w^*)$.
\begin{corollary}\label{co:swing}
Consider the candidate Lyapunov function $$V(\delta,\w)=V_k(\w)+V_p(\delta)$$where $$V_k(\w)=\frac{1}{2}M(\w-\bar{\w}_s)^2$$and$$V_p(\delta)=\gamma[-\cos\delta+\cos\bar{\delta}-(\delta-\bar{\delta})\sin\bar{\delta})]\;.$$ Assume that $\frac{P_m}{\gamma}<\frac{2}{\pi}$ and let $c:=V_p(\frac{\pi}{2})$. Then the solutions of the SMIB system described by the swing equation \eqref{swingib} starting from any initial condition in the set $\Omega = \{(\delta,\w) \in [-\pi,\pi]\times\R:\,V(\delta,\w)\leq c\}$ converge asymptotically to the equilibrium $(\delta,\w)=(\mathrm{arcsin} \frac{P_m}{\gamma},\w^*)$.
\end{corollary}
\begin{proof}
The proof is similar to the proof of the Theorem \ref{ThImproved} and therefore omitted. Note that here, $\dot{V}=-A(\w-\w^*)^2$ and hence here the condition $V<V_p(\frac{\pi}{2})$ suffices for proving the stability.
\end{proof}
Corollary \ref{co:swing} and Theorem \ref{ThImproved} characterize similar region of attraction estimates for both models (improved swing and the conventional one) if $c_k\geq c_p$ (and as a result $c=c_p$).
\section{Simulation}
In this section we provide simulations verifying our results. Examples depict the mismatch between the behavior suggested by the swing equation and the improved version. In all simulations, the parameters are set as follows: $M=0.2$, $A=0.04$, $\w^*=(2\pi)60$, and $\gamma=2$. Note that $J=\frac{M}{w^*}$ and $D_d=\frac{A}{\w^*}$.
\subsection{Constant Load: Example 1}
As described in Section III, the steady state value of the frequency differs for the swing equation and the improved swing model under similar constant loads. Figure \ref{fig:0}(a) illustrates this issue. In this example, we set $P_m=1\pu$ and  $P_e=2\pu$, and start from the initial condition $f(0)=60\mathrm{Hz}$. The steady state value of the frequency is $56.02\mathrm{Hz}$ for the swing equation and $55.72\mathrm{Hz}$ for the improved swing model.
\subsection{Constant Load: Example 2}
An estimate of the region of attraction for the improved swing equation is provided in Section III ($\Omega_s$ in Theorem \ref{th:stability}). Figure \ref{fig:0}(b) illustrates a solution that initiates just out of the domain and becomes unstable. However, the conventional swing equation is falsely depicting that the system remains stable. Here, we set $P_m=1 \pu$ and $P_e=4.65\pu$, and start from $f(0)=24\mathrm{Hz}$. According to \eqref{bounded}, the estimate of region of attraction allows for $f>25\mathrm{Hz}$.
\subsection{Constant Load: Example 3}
According to the analysis in Section III, condition \eqref{delta} should hold for the improved swing equation, so that there exist a steady state frequency value. Figure \ref{fig:0}(c) shows that the system becomes unstable if the inequality \eqref{delta} is violated. In this example we adjust $P_m=1\pu$ and $P_e=4.90\pu$ such that $\Delta$ as defined in \eqref{e:Delta} possesses a negative value.
\subsection{SMIB: Different Behavior}
The behavior of both systems, conventional and improved, are similar when connected to an infinite bus. However still with some specific initial conditions, the systems act quite differently. Figure \ref{fig:0}(d) illustrates an example of this different behavior. 
\subsection{SMIB: Region of Attraction}
Figure \ref{fig:roa} illustrates the phase portrait of the system \eqref{swing4} and the Lyapunov function $V(\delta,\w)$ level sets. It is verified that our estimate of the domain of attraction is not very conservative. Note that there are solutions outside the estimate of region of attraction that still converge, however the solutions further away from the estimate of domain of attraction diverge from the equilibrium.
\begin{figure}
	\centering
	\includegraphics[width=8.7cm]{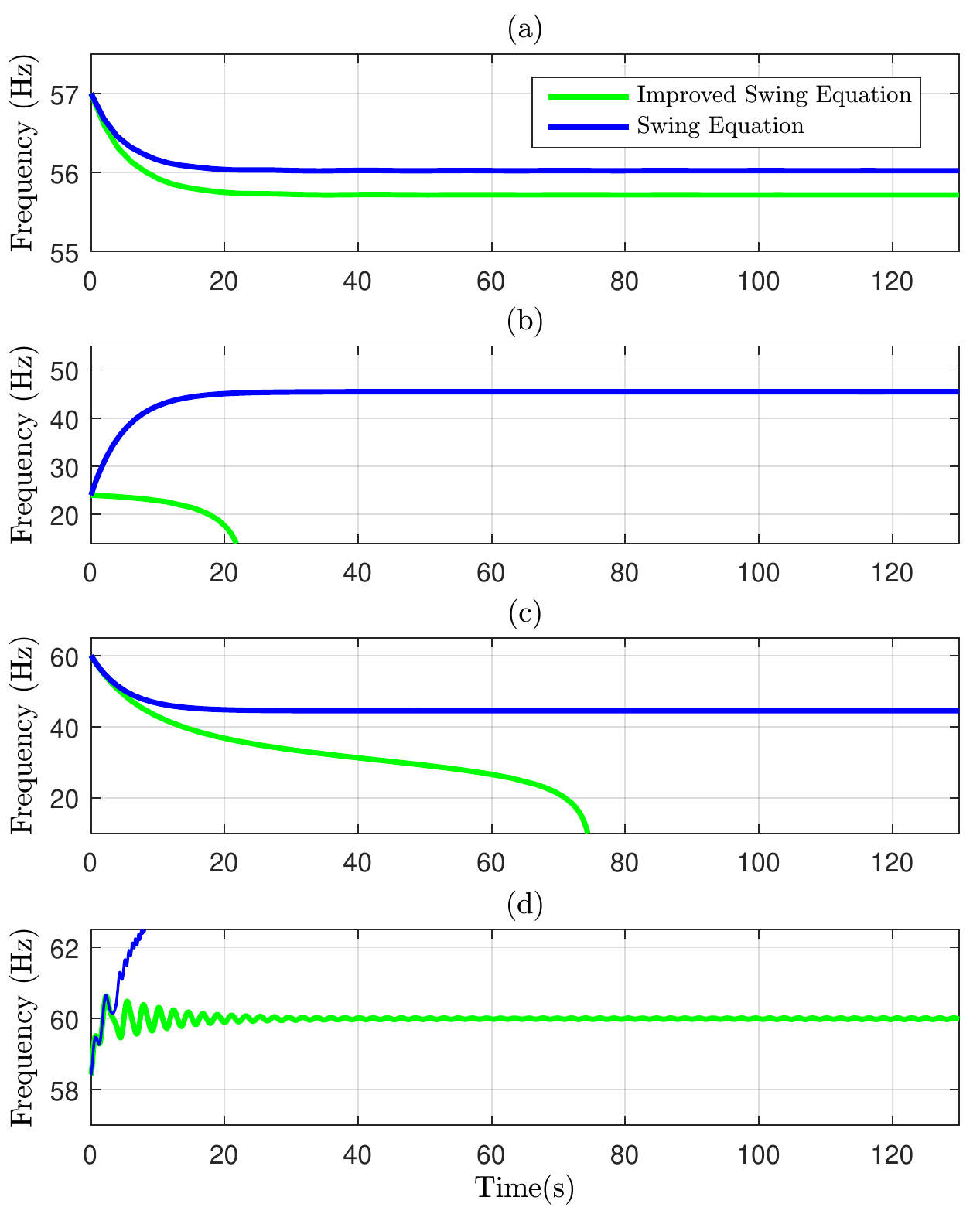}
	\caption{Simulation results for the case with constant load (a-c) and SMIB (d) show different behavior of the two models: (a) The models reach different steady-state frequency value while bearing the same load.  (b) A solution of the improved swing equation starting from outside of the region of attraction estimate, diverges from the equilibrium. Nevertheless, the solution of the conventional swing equation converges.  (c) The frequency in the improved swing equation is not stable when condition \eqref{delta} does not hold. The stability of the swing equation does not necessitate such a condition.  (d) While the conventional swing SMIB model \eqref{swingib} suggests divergence from the equilibrium, the frequency of the improved model converges to the desired value. Note that here, the initial condition is outside the region of attraction estimate for both models.}
	\label{fig:0}
\end{figure} 
\begin{figure}
	\centering
	\includegraphics[width=8.7cm]{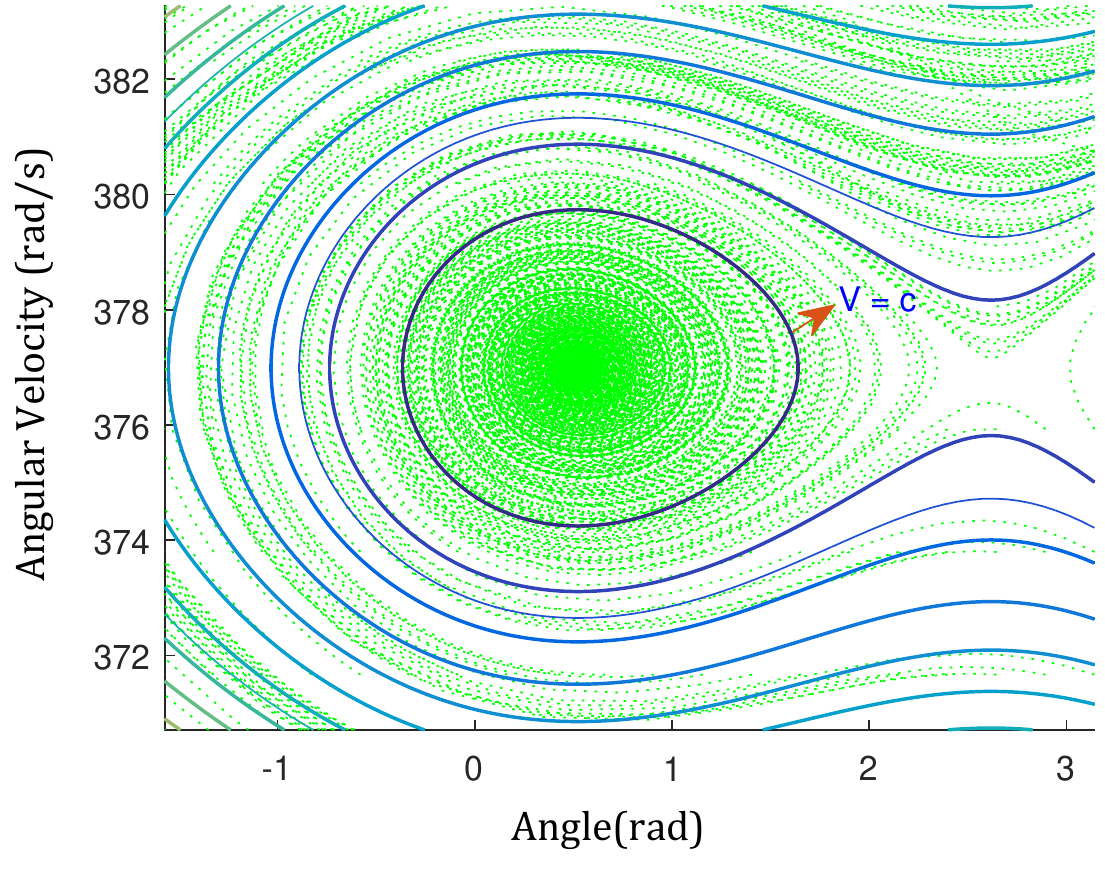}
	\caption{The estimate of the domain of attraction for system \eqref{swing4} (SIMB) is verified. All solutions (dotted green) within the Lyapunov level set $\{V(\delta, \w)=c\}$ (dark blue) converge to the equilibrium $(\delta,\w)=\big(\frac{\pi}{6},(2\pi)60\big)$. Note that here $c_p<c_k$, and hence $c=c_p=V_p(\frac{\pi}{2})$.}
	\label{fig:roa}
\end{figure}
\section{Conclusion and Future Work}
We have investigated the properties of an improved swing equation without relying on linearization. Modeling the synchronous generator by this equation, two scenarios are considered in this paper. First, the stability of a single generator connected to a constant load is proved and frequency regulation through a proposed controller is achieved. In the second scenario, the synchronous machine is connected to an infinite bus. As a contribution with respect to \cite{Zhou2009}, where similar dynamics are investigated through linearization, here a nonlinear Lyapunov analysis is provided to prove stability and frequency regulation. Finally, simulations are carried out to show that the swing equation model gives rise to a behavior that does not match what is suggested by the improved swing equation. Future works include considering voltage dynamics and multi-machine systems.
\bibliographystyle{ieeetran}
\bibliography{ref}

\begin{thebibliography}{10}
\providecommand{\url}[1]{#1}
\csname url@samestyle\endcsname
\providecommand{\newblock}{\relax}
\providecommand{\bibinfo}[2]{#2}
\providecommand{\BIBentrySTDinterwordspacing}{\spaceskip=0pt\relax}
\providecommand{\BIBentryALTinterwordstretchfactor}{4}
\providecommand{\BIBentryALTinterwordspacing}{\spaceskip=\fontdimen2\font plus
\BIBentryALTinterwordstretchfactor\fontdimen3\font minus
  \fontdimen4\font\relax}
\providecommand{\BIBforeignlanguage}[2]{{%
\expandafter\ifx\csname l@#1\endcsname\relax
\typeout{** WARNING: IEEEtran.bst: No hyphenation pattern has been}%
\typeout{** loaded for the language `#1'. Using the pattern for}%
\typeout{** the default language instead.}%
\else
\language=\csname l@#1\endcsname
\fi
#2}}
\providecommand{\BIBdecl}{\relax}
\BIBdecl

\bibitem{Shahidehpour2013}
M.~Shahidehpour and M.~Khodayar, ``Cutting campus energy costs with
  hierarchical control: The economical and reliable operation of a microgrid,''
  \emph{Electrification Magazine, IEEE}, vol.~1, no.~1, pp. 40--56, Sept 2013.

\bibitem{Zhou2009}
J.~Zhou and Y.~Ohsawa, ``Improved swing equation and its properties in
  synchronous generators,'' \emph{Circuits and Systems I: Regular Papers, IEEE
  Transactions on}, vol.~56, no.~1, pp. 200--209, Jan 2009.

\bibitem{Arjan2013}
S.~Fiaz, D.~Zonetti, R.~Ortega, J.~Scherpen, and A.~van~der Schaft, ``A
  port-hamiltonian approach to power network modeling and analysis,''
  \emph{European Journal of Control}, vol.~19, no.~6, pp. 477 -- 485, 2013,
  lagrangian and Hamiltonian Methods for Modelling and Control.

\bibitem{Caliskan2014}
S.~Caliskan and P.~Tabuada, ``Compositional transient stability analysis of
  multimachine power networks,'' \emph{Control of Network Systems, IEEE
  Transactions on}, vol.~1, no.~1, pp. 4--14, March 2014.

\bibitem{Tabuada2015}
S.~Y. Caliskan and P.~Tabuada, ``Uses and abuses of the swing equation model,''
  in \emph{Decision and Control (CDC), 2015 IEEE 54th Annual Conference on},
  Dec 2015, pp. 6662--6667.

\bibitem{Weiss20141}
V.~Natarajan and G.~Weiss, ``Almost global asymptotic stability of a constant
  field current synchronous machine connected to an infinite bus,'' in
  \emph{Decision and Control (CDC), 2014 IEEE 53rd Annual Conference on}, Dec
  2014, pp. 3272--3279.

\bibitem{Weiss20142}
V.~\vspace{0mm}Natarajan and G.~Weiss, ``A method for proving the global
  stability of a synchronous generator connected to an infinite bus,'' in
  \emph{Electrical Electronics Engineers in Israel (IEEEI), 2014 IEEE 28th
  Convention of}, Dec 2014, pp. 1--5.

\bibitem{Tjerk2016}
A.~J. van~der Schaft and T.~Stegink, ``Perspectives in modeling for control of
  power networks,'' \emph{Annual Reviews in Control}, Spring Issue 2016, to
  appear.

\bibitem{Machowski2009}
J.~Machowski, J.~Bialek, and J.~Bumby, \emph{Power System Dynamics: Stability
  and Control}, 2nd~ed.\hskip 1em plus 0.5em minus 0.4em\relax Wiley, 2008.

\bibitem{Zhao1995}
Z.~Zhao, F.~Zheng, J.~Gao, and L.~Xu, ``A dynamic on-line parameter
  identification and full-scale system experimental verification for large
  synchronous machines,'' \emph{IEEE Transactions on Energy Conversion},
  vol.~10, no.~3, pp. 392--398, Sep 1995.

\end{thebibliography}
\IEEEpeerreviewmaketitle
\end{document}